\documentclass[reprint,prx,floatfix,nofootinbib,superscriptaddress]{revtex4-1}
\usepackage{color}
\usepackage[utf8]{inputenc}
\usepackage[T1]{fontenc}

\usepackage[colorlinks=true]{hyperref}
\usepackage{graphicx}
\usepackage{dcolumn}
\usepackage{bm}
\usepackage{amsmath,amssymb,amsthm} 
\usepackage{amsfonts}
\usepackage{comment}
\usepackage[usenames,dvipsnames]{pstricks}
\usepackage{subfigure}
\usepackage{braket}
\usepackage{float}
\usepackage{xcolor}
\usepackage{todonotes}
\newtheorem{proposition}{Proposition}

\begin{document}

\title{Connecting Hodge and Sakaguchi-Kuramoto: a mathematical framework for coupled oscillators on simplicial complexes}

\author{Alexis Arnaudon}
\affiliation{Department of Mathematics, Imperial College, London SW7 2AZ, UK}
\affiliation{Blue Brain Project, École Polytechnique Fédérale de Lausanne (EPFL), Campus Biotech, 1202 Geneva, Switzerland}
\author{Robert L. Peach}
\affiliation{Department of Mathematics, Imperial College, London SW7 2AZ, UK}
\affiliation{Department of Neurology, University Hospital Würzburg, Würzburg, Germany}
\author{Giovanni Petri}
\affiliation{ISI Foundation, via Chisola 5, Turin, Italy}
\affiliation{ISI Global Science Foundation, 33 W 42nd St, 10036 New York NY, USA}
\author{Paul Expert}
\affiliation{Global Business School for Health, University College, London, WC1E 6BT, UK}
\affiliation{Department of Primary Care and Public Health, Imperial College, London SW7 2AZ, UK}
\affiliation{World Research Hub Initiative, Tokyo Institute of Technology, Tokyo, JP}

\begin{abstract}

We formulate a general Kuramoto model on weighted simplicial complexes where phases oscillators are supported on simplices of any order $k$. 
Crucially, we introduce linear and non-linear frustration terms that are independent of the orientation of the $k+1$ simplices, providing a natural generalization of the Sakaguchi-Kuramoto model. 
In turn, this provides a generalized formulation of the Kuramoto higher-order parameter as a potential function to write the dynamics as a gradient flow.
With a selection of simplicial complexes of increasingly complex structure, we study the properties of the dynamics of the simplicial Sakaguchi-Kuramoto model with oscillators on edges to highlight the complexity of dynamical behaviors emerging from even simple simplicial complexes. 
We place ourselves in the case where the vector of internal frequencies of the edge oscillators lies in the kernel of the Hodge Laplacian, or vanishing linear frustration, and, using the Hodge decomposition of the solution, we understand how the nonlinear frustration couples the dynamics in orthogonal subspaces.
We discover various dynamical phenomena, such as the partial loss of synchronization in subspaces aligned with the Hodge subspaces and the emergence of simplicial phase re-locking in regimes of high frustration.

\end{abstract}

\maketitle

\section{Introduction}

Synchronisation is an ubiquitous phenomenon observed in many complex systems across spatial and temporal scales~\cite{Arenas:2008ku}, from the firing patterns of neurons and the communication of fireflies, to the flow of traffic~\cite{Petri:2013bz,OKeefe:1971bj,Hafting:2005dp}.
One of the most popular dynamical systems, capable of reproducing a wide range of observed synchronisation behaviours, is the Kuramoto model of coupled oscillators
~\cite{kuramoto1975self,acebron2005kuramoto,rodrigues2016kuramoto}.
Whilst the model was originally formulated in terms of all-to-all interacting oscillators, the interactions between oscillators are commonly considered inhomogeneous and represented with a graph, whose structure affects the resulting dynamics. 
For example, while the full synchronisation of the oscillator population is usually a strong attractor for the dynamics irrespective of the underlying graph ~\cite{Arenas:2008ku,rodrigues2016kuramoto}, the transient dynamics on the path towards synchronisation can reveal the modular structure of the oscillators' interactions~\cite{Arenas:2006ba}.

Beyond the structure of oscillator interactions, other variations of the Kuramoto model have been studied extensively, including: time-delayed interactions~\cite{yeung1999time,Hellyer:2015ci}, oriented or signed interactions~\cite{hong2011kuramoto,delabays2019kuramoto}, time-varying parameters, stochasticity, and more, see~\cite{Arenas:2008ku} for a comprehensive review.
Of particular interest for this study, the introduction of a frustration parameter~\cite{sakaguchi1986soluble} in the nonlinear term of the Kuramoto model, then known as the Sakaguchi-Kuramoto model, can produce rich dynamics~\cite{Abrams:2004hq,Shanahan:2010go,omel2012nonuniversal,nicosia2013remote} and appears in many applications~\cite{wiesenfeld1996synchronization,filatrella2008analysis}. 

Recently, the study of higher-order interactions between elements of a system, that is, models with interactions involving more than two nodes, has garnered momentum and interest~\cite{Battiston:2020kp}.
Higher-order interactions are typically represented with hypergraphs or simplicial complexes, both of which generalize the graph representation of pairwise interactions to instead encode three-, four- and higher-way interactions. 
Naturally, extensions of well known dynamical systems have been proposed to investigate the effect of higher order interactions on their behavior~\cite{iacopini2019simplicial,Carletti:2020ux,schaub2020random, millan2020explosive, DeVille_2021,Ghorbanchian_2021}. 

The Kuramoto model --being a paradigmatic model for synchronization phenomena-- is no exception. 
In this case, however, there are two main avenues to extend classical oscillator models to higher-order. 
The first approach maintains the usual setup of phases defined on nodes of a systems and upgrades the interactions to the polyadic case, e.g. using simplicial complexes as the underlying connectivity structure.
Recent works investigated variations of this node Kuramoto model with higher-order interactions introducing various types of coupling terms~\cite{Skardal:2019ik,Skardal:2020fl}.
These models display a rich variety of synchronization and desynchronization phenomena, as well as multi-stable behavior.
The second approach instead promotes phase variables from nodes to higher-order simplices, thus defining phases for edges, triangles, and all higher-order interactions, coupled by boundary operators as generalized incidence matrices.
Pioneering work in this direction, \cite{millan2020explosive} showed that the edge dynamics projected onto the nodes and faces possesses explosive synchronization properties when specific nonlinear and non-local couplings are introduced between the two projections. More recently, also a version of the same model with local coupling between orders was introduced~\cite{calmon2021topological}.

In this paper, we extend this latter \textit{simplicial} Kuramoto model~\cite{millan2020explosive} to include: i) weights on any simplices with a precise mathematical formulation based on discrete differential geometry; and, more importantly, ii) linear and nonlinear frustrations. 
We will refer to the former --\textit{linear}-- frustrations as natural frequencies yielding non-fully synchronized stationary states, and to the latter --\textit{non-linear}-- as the higher-order generalization of the Sakaguchi-Kuramoto model~\cite{sakaguchi1986soluble}. 
The difficulty of introducing proper nonlinear frustration comes from the orientation of the simplices which make, even a naive frustration, orientation dependent. 
Here, inspired by previous work on higher-order random walks~\cite{schaub2020random}, we lift the simplices to double their numbers with opposite signs, obtaining an equivalent formulation without frustration and an orientation independent frustration.
We then study the resulting frustrated simplicial Kuramoto model on edges with numerical simulations of oscillators which internal frequency vector lies in the kernel of the Hodge Laplacian, using several measures to quantify the type of dynamics, such as Hodge decomposition, the order parameter and the largest Lyapunov exponent. 
models of identical oscillators which internal frequency vector lies in the kernel of the Hodge Laplacian

\section{Theory}
\subsection{Simplicial complexes and Hodge Laplacian}\label{sec:sc_hl}

The central elements of the mathematical formulation of the Kuramoto model on simplicial complexes are the boundary operators and the related Hodge Laplacians, which are, respectively, generalizations to higher order structures of the graph incidence matrices and of the Laplacian operator. 
We briefly review the main concepts we will use in our work following~\cite{grady2010discrete}, see also \cite{DeVille_2021}, with additional details in Appendix~\ref{discrete_geo} . 
A $k$-simplex is defined by a set of $k+1$ nodes (a 1-simplex is an edge, a 2-simplex is a triangle, etc.).
A simplicial complex is defined as a set of simplices in which every face of a simplex is also a simplex.
For our purposes, the relevant connectivity between $k-$simplices will be that induced by sharing a $(k-1)$-simplex as a face, e.g. triangles  sharing an edge, or by being faces of a $(k+1)$-simplex, e.g. edges belonging to the same triangle.
A $k$-chain within a simplicial complex is a linear combination of $k$-simplices.
We denote by $n_k$ the number of $k$-simplices of a complex, which is also the dimension of the $k$-chains and $k$-cochains vector spaces, dual to $k$-chains.
The coboundary operator $N_k$ and its dual $N_k^*$ on a simplicial complex are defined using the generalized incidence matrices $B_k^T\in M^{n_{k}\times n_{k+1}}$ which encode the topology of a simplicial complex, and the weight matrices $W_k$, which are diagonal matrices of the  $k$-simplices weights
\begin{align}
    N_k = B_k\, , \qquad N_k^* = W_{k} B_k^T W^{-1}_{k+1}\, . 
\end{align}
The weight matrices $W_k$ can be chosen in an \textit{ad-hoc} fashion and no formal relations need to exist between the different order $k$. The only relative constraint is for the weights to be positive in order to remain in the realm of unsigned graphs. Note that our notation follows~\cite{grady2010discrete} which differs from the convention commonly used for these operators.
Both act on $k$-cochains, defined as linear functional on the space of $k$-chains, see Appendix~\ref{discrete_geo}.
The Hodge Laplacian of order $k$ can then be written as 
\begin{align}
    L_k &= L_k^{down}+L_k^{up} \label{eq:up_down_Lk} \\
    &:=N_{k-1} N_{k-1}^* + N_k^* N_k\, . \label{eq:weighted_Lk}
\end{align}
For $k=0$, $W_0 = I$ and $W_1 = I$, we obtain the graph Laplacian $L_0=D-A$ with $A$ the simplicial complex 1-skeleton, namely the graph node adjacency matrix, and $D$ the diagonal matrix of the nodes degree. The choice $W_0 = D^{-1}$ defines the normalized graph Laplacian $L_0^{norm} = I - D^{-1}A$.

The graph Laplacian $L_0$ can produce two types of dynamics. 
When acting on the left of a distribution $f$, it yields the consensus dynamics $\dot f = L_0 f$ for any choice of $W_1$ while by acting on the right, it corresponds to the diffusion dynamics $\dot p = pL_0$.
Equally, both types of dynamics also exist for the edge Laplacian $L_1$~\cite{muhammad2006control,schaub2020random}, defined as
\begin{align}
    L_1 = B_0 W_0 B_0^T W^{-1}_1 + W_1 B_1^T W_2^{-1} B_1 \, .
    \label{eq:weighted_L1}
\end{align}
We refer to Appendix~\ref{diffusion_kuramoto} for the diffusion formulation of the weighted simplicial Kuramoto model. For the remainder of this paper we will use the standard consensus formulation, but we emphasise that our formulation is not restricted to consensus dynamics.

\subsection{Simplicial Kuramoto model}\label{sec:simplicial_kuramoto}

The Kuramoto model~\cite{kuramoto1975self} is typically formulated for a node phase dynamical variable $\theta \in \mathbb R^{n_0}$, with natural frequencies $\omega = (\omega_1, \ldots, \omega_{n_0}) \in \mathbb R^{n_0}$ that are sitting on the nodes of a graph $G = (V, E)$ ($|V|=n_0$, $|E|=n_1$) and interact through the graph adjacency matrix $A_{ij}\in \mathbb R^{n_0\times n_0}$
\begin{align}
    \dot \theta_i = \omega_i - \sigma \sum_j A_{ij} \sin(\theta_i - \theta_j)\, .
    \label{eq:kuramoto}
\end{align}

For simplicity, we will consider a unit coupling $\sigma=1$ throughout the remainder of this paper and will thus omit it from here on. 
The unweighted node Kuramoto model can be equivalently formulated in vector form using the $n_0 \times n_{1}$ incidence matrix $B_0^T$~\cite{jadbabaie2004stability} and a vector of internal frequencies $\omega$ as
\begin{align}
    \dot \theta = \omega - B_0^T\sin(B_0\theta)\, , 
    \label{eq:incidence_kuramoto}
\end{align}
which is approximated by the Laplacian dynamics $\dot \theta = \omega - B_0^TB_0\theta = \omega - L_0\theta$ in the limit $B_0\theta\ll 1$. 
When $\omega_i=\omega$ for all $i$, it is customary to study the node Kuramoto model in a frame rotating at $\omega t$ and thus ignore the internal frequencies, yielding $\dot \theta = L_0\theta$. 
The Kuramoto model is therefore a nonlinear extension of the consensus dynamics introduced in Section~\ref{sec:sc_hl}. 

The weighted simplicial Kuramoto model is then given for a time-dependent $k$-cochain $\theta^{(k)}$, see~\cite{millan2020explosive,DeVille_2021} for the original equations, as
\begin{align}
    \dot \theta^{(k)} = -N_{k-1} \mathrm{sin}\left (N_{k-1}^*\theta^{(k)}\right) - N_k^* \mathrm{sin}\left(N_k\theta^{(k)}\right)\, ,
    \label{eq:simplicial_kuramoto}
\end{align}
or equally with the weight and incidence matrices
\begin{align}
    \dot \theta^{(k)} &= -B_{k-1} \mathrm{sin}\left (W_{k-1} B_{k-1}^T W^{-1}_k\theta^{(k)}\right)\nonumber  \\
    &\qquad - W_k B_k^T W_{k+1}^{-1} \mathrm{sin}\left (B_k\theta^{(k)}\right ) \,.\label{eq:weighted_simplicial_kuramoto}
\end{align}
We emphasize that the positions of the weight matrices are not arbitrary but constrained by the geometrical nature of the coboundary operators, see Appendix~\ref{discrete_geo}. The definition and interpretation of the weights themselves are defined by the system under study (which may include geometrical constraints). The weighted model can be seen as an extension of the Kuramoto model on weighted graphs, where the weights represent heterogeneous couplings. The weights on simplices of different order can be coupled, but this is not a necessary requirement and each order can capture independent characteristics of the system studied. We briefly explore the effect of weights on the dynamics of the Sakaguchi-Kuramoto model that forms the basis of our numerical experiments and is introduced in the next section in~\ref{sec:weighted}.
In the limit where $\theta$ is close to the subspace $\mathrm{ker}(L_k)$, we recover the linear consensus dynamics $\dot\theta^{(k)}=L_k\theta^{(k)}$. 
For $k=0$ and a connected graph, the kernel subspace consists of a constant vector, or full synchronization.

Similarly to the node Kuramoto, the internal frequencies of the oscillators can be introduced via a change of rotating frame $\theta^{(k)} \to \theta^{(k)} - h^{(k)} t$ for any vector $h^{(k)} \in \mathrm{ker}(L_k)$.
Indeed, such a vector will leave invariant the nonlinear terms, due to the presence of the boundary operator, and thus only adds a constant drift to the phases. 
Again, if we consider $k=0$ and a connected graph, the kernel of $L_0$ is the constant vector, corresponding to the stationary state of consensus dynamics, and thus, by extension, the node Kuramoto model in full synchronization.
For higher-order Kuramoto models $k>0$, the dimension of the  $\mathrm{ker}(L_k)$ corresponds to the number of $k$-dimensional holes, i.e. holes bounded by $k$-simplices, or --equivalently-- to the Betti number $\beta_k$ of the simplicial complex. 
For $k=0$, the Betti number $\beta_0$ corresponds to the number of connected components, and the stationary states are given by the piece-wise constant vectors to which each component will synchronize.
We did not introduce by hand any internal frequencies at this stage, as we will see in remainder of this section that they naturally emerge as a form of linear frustration.

\subsection{Simplicial Sakaguchi-Kuramoto model}\label{sec:frustrated_kuramoto}

The frustration in the Kuramoto model was first introduced in the Kuramoto–Sakaguchi model~\cite{sakaguchi1986soluble}, and has been studied in the context of graph theory, where the graph topology can give rise to rich repertoires of stationary states such as chimera states ~\cite{Abrams:2004hq,Shanahan:2010go} and remote synchronization~\cite{nicosia2013remote}.
The frustrated node Kuramoto model is usually written as
\begin{align}
    \dot \theta_i = \omega_i - \sum_j A_{ij} \sin(\theta_i - \theta_j + \alpha_{ij})\, ,
    \label{eq:frustrated_kuramoto}
\end{align}
where $\alpha \in \mathbb R^{n_1}$ is the edge frustration vector, often taken to be constant $\alpha_{ij} = \alpha_1$.
This equation cannot be directly formulated using the incidence matrices because the relative sign between the difference of phases $\theta_i-\theta_j$ and $\alpha_{ij}$ must be independent from the orientation of edges. 
In the adjacency matrix formulation, the orientation of edges is `hidden', because $L_0=B_0^TB_0$ and $A=D-L_0$ are independent of edge orientation, and the choice of ordering $\theta_i-\theta_j$, instead of $\theta_j-\theta_i$, is possible irrespective of the edge orientation. 
If one writes $B_0^T \sin(B_0\theta + \alpha_1)$, the resulting order in the difference of phases depends on the choice of edge orientation and will not be `node-centered', i.e. the $\theta_i$ term will not always appear in front.

Nevertheless, it is possible to introduce a frustration in the general formulation of the Kuramoto model~\eqref{eq:simplicial_kuramoto} with coboundary operators such that it reduces to the frustrated Kuramoto~\eqref{eq:frustrated_kuramoto} for $k=0$ and remains orientation invariant for $k+1$ simplices.
Our construction uses two ingredients: i) lift matrices~\cite{schaub2020random}, defined as
\begin{align}
    V_k = 
    \begin{pmatrix}
    I_{n_k}\\
    -I_{n_k}
    \end{pmatrix}\, , 
\end{align}
for any order $k$, and ii) the projection onto the positive or negative entries of any matrix $X$, defined element-wise as
\begin{align}
    X_{ij}^\pm = \frac12 \big(X_{ij} \pm \left |X_{ij}\right |\big )\,\quad \forall ij, 
    \label{pm-projection} 
\end{align}
where $|\cdot|$ denotes the absolute value function. The lift matrices create duplicates of simplices of order $k$ with an orientation opposite to the original one, whilst the projection sets half of the doubled simplices to zero, i.e. removes them, based on their signs.
One can define the lift of the coboundary operator as
\begin{align}
    N_k \to V_{k+1} N_k V_k^T\, . 
    \label{N_lift}
\end{align}
The projection to positive or negative entries is often used to define directed node graph Laplacians~\cite{grady2010discrete,chapman2015advection} by transforming the edge orientation to an edge direction with either
\begin{align}
    L_{0,\mathrm{out}} = N_0^* N_0^+\, , \quad \mathrm{or} \quad  L_{0, \mathrm{in}} = (N_0^-)^* N_0\, .
\end{align}
With $D_\mathrm{out/in}$ the diagonal matrices of out- or in- degrees and $A_\mathrm{dir}$ the corresponding directed adjacency matrix, $L_{0,out}$ models the directed diffusion dynamics written explicitly as $L_{0,\mathrm{out}} = D_\mathrm{out} - A_\mathrm{dir}$ and $L_{0,in}$ corresponds to the directed consensus dynamics with $L_{0,\mathrm{in}} = D_\mathrm{in} - A_\mathrm{dir}$. 

As we have seen in the construction of the weighted simplicial Kuramoto model~\eqref{eq:simplicial_kuramoto}, we are using the formulation that yields consensus dynamics. We will thus consider the associated projection onto the negative entries of the lifted simplicial Laplacian as
\begin{align}
    \widehat L_k = N_{k-1}^-N_{k-1}^* + (N_k^*V_{k+1}^T)^-V_{k+1}N_k\, .
    \label{laplacian-hat}
\end{align}
First, we note that $\widehat L_k=L_k$, see Appendix~\ref{lift_proj}, thus the application of the lift and the projection has a trivial effect on the Hodge Laplacian, but crucially it allows us to introduce the frustration via the linear frustration operator
\begin{align}
   \mathcal F_k^{\alpha_k}(N_k): x \mapsto N_k x + \alpha_k\, ,
\end{align}
acting on any cochain $x$ and arbitrary frustration cochain $\alpha_k$.
We can now formulate the frustrated simplicial Kuramoto model as
\begin{align}
    \dot \theta^{(k)} &= - \mathcal F_k^{\alpha_k}(N_{k-1}) \left[\mathrm{sin}\left (N_{k-1}^*\theta^{(k)}\right)\right] \nonumber \\
    &\qquad - (N_k^*V_{k+1}^T)^- \mathrm{sin}\left(\mathcal F_{k+1}^{\alpha_{k+1}}( V_{k+1}N_k)\left[\theta^{(k)}\right]\right)\\
    &=-\alpha_k - N_{k-1}\mathrm{sin}\left ( N_{k-1}^* \theta^{(k)}\right) \nonumber \\
    &\qquad - (N_k^*V_{k+1})^- \mathrm{sin}\left (V_{k+1} N_k \theta^{(k)} + \alpha_{k+1}\right)\, .
    \label{frustrated_simplicial_kuramoto}
\end{align}
By construction, our formulation is independent of the orientation of the $k+1$-simplices but not of the $k$-simplices because only the action of the $k+1$ lift is non trivial as it acts inside the nonlinear part of the equation, see Appendix~\ref{lift_proj}.

From this point of view, $\alpha_k$ is a linear frustration, whilst $\alpha_{k+1}$ is a nonlinear frustration.
Like in \ref{sec:simplicial_kuramoto} where the internal frequencies are all equal, $\alpha_k$ can be an arbitrary vector not necessarily in $\mathrm{ker}(L_k)$, which corresponds to equal internal frequencies in the node Kuramoto.
This can lead to a variety of dynamics, including partially synchronized dynamics or even non-stationary dynamics if its amplitude is large enough~\cite{millan2020explosive}.

For $k=0$, we recover the frustrated node Kuramoto model~\eqref{eq:frustrated_kuramoto} as
\begin{align}
    \dot \theta^{(0)} = -\alpha_0 - (N_0^* V_1^T)^- \mathrm{sin}\left (V_1 N_0 \theta^{(0)} + \alpha_1\right)\, , 
\end{align}
where the natural frequencies vector $\alpha_0$ naturally appears from the frustration operator, whilst the rest vanishes with $N_{-1}=0$.
The case where $k=1$ constitutes the main equation we consider in the rest of this paper, namely the frustrated edge simplicial Kuramoto model
\begin{align}
    \dot \theta^{(1)} &= -\alpha_1 - N_0\mathrm{sin}\left ( N_0^* \theta^{(1)}\right) \nonumber \\
    &-  (N_1^*V_2)^- \mathrm{sin}\left (V_2 N_1 \theta^{(1)} + \alpha_2\right)\, ,
    \label{eq:edge_frustrated_kuramoto}
\end{align}
which is invariant under change of face orientations, but not under change of edge orientation.

\subsection{Hodge decomposition of the dynamics}\label{sec:hodge_decomposition}

The Hodge decomposition is an important tool to study the properties of simplicial complexes. 
Here, we use it to decompose the dynamics of the oscillators on the simplicial Kuramoto model to understand their properties in relation to the amount of frustration applied.
The Hodge decomposition theorem states that the space of $k$-cochains can be decomposed into three orthogonal spaces~\cite{Eckmann:1945,jiang2011statistical}
\begin{align}
    C^{(k)}=\mathrm{Im}(N_{k-1})\oplus \ker(L_k)\oplus \mathrm{Im}(N_k^*)\, ,
\end{align}
which can be seen as analogues to the gradient, harmonic and curl space respectively. When $k=1$ the three orthogonal spaces are exactly the gradient, harmonic and curl space respectively. 
Any $k$-cochain $\theta^{(k)}$ can thus be projected onto each subspace $\theta^{(k)} = \theta_\mathrm{g}^{(k)} + \theta_\mathrm{h}^{(k)} + \theta_\mathrm{c}^{(k)}$ as follow
\begin{align}
\begin{split}
    \theta_\mathrm{g}^{(k)} &= N_k \theta^{(k-1)} \\
    L_k\theta_\mathrm{h}^{(k)} &= 0\\
    \theta_\mathrm{c}^{(k)} &= N_{k+1}^* \theta^{(k+1)}\, .
\end{split}
\label{hodge_projections}
\end{align}
where $\theta^{(k-1)}$ and $\theta^{(k+1)}$ are the corresponding potentials.
Here, instead of computing these potentials, as done for example in~\cite{millan2020explosive}, we project the $k$-cochain $\theta^{(k)}$ onto each subspace using the projection operators
\begin{align}
    \begin{split}
    P_\mathrm{grad} &= p_{\mathrm{grad}}^Tp_{\mathrm{grad}}\\
    P_\mathrm{curl} &= p_{\mathrm{curl}}^Tp_{\mathrm{curl}}\\
    P_\mathrm{harm} &= p_{\mathrm{harm}}^T p_{\mathrm{harm}}\, , 
    \end{split}
    \label{hodge_projections}
\end{align}
where the matrices $p_\mathrm{grad}$ and $p_\mathrm{curl}$ are the orthonormal bases of the ranges of $N_k$ and $N_{k+1}^*$ and $p_\mathrm{harm}$ the orthonormal basis of the kernel of $L_k$.

\subsection{Simplicial order parameter}\label{sec:SOP}

Probably the most popular and fundamental tool to measure the level of synchronization in a coupled dynamical system is the order parameter. It is usually defined as
\begin{align}
    R_{0, c}^2(\theta) :=\frac{1}{n_0} \left | \sum_{i=1}^{n_0}
    \exp\left( j\theta_i\right) \right |^2\, ,
    \label{eq:node_OP}
\end{align}
where $j=\sqrt{-1}$, and was introduced for the original Kuramoto model on a complete graph, i.e. where all oscillators are coupled. 
The generalization of the order parameter to any graph structure~\cite{jadbabaie2004stability} can be expressed as
\begin{align}
    R_{0,g}^2(\theta) := 1 + \frac{2}{n_0^2} 1_{n_1}\cdot( \cos(N_1^* \theta) -1 )\, , 
    \label{order_naive}
\end{align}
where $1_{n_1}$ is the unit vector of dimension $n_1$, see Appendix~\ref{simplicial_order} for the details.
This formulation allows one to write the node Kuramoto model with uniform natural frequencies as a gradient flow of the form
\begin{align}
    \dot \theta = \frac12 n_0^2 \nabla_\theta R_0^2(\theta)\, . 
\end{align}
Notice that the usual minus sign is not needed as the order parameter is a concave function.
Notice that only the cosine term is needed to express the gradient flow, while the other constant terms are needed for the normalization.

For a simpler derivation, we will thus modify the normalization to define the simplicial order parameter (SOP) as
\begin{align}
    R_k^2(\theta^{(k)}) = \frac{1}{C_k}\Big( 1_{n_{k-1}} \cdot W_{k-1}^{-1} \cos(N_{k-1}^*\theta^{(k)})\Big. \nonumber \\
    \Big.+1_{n_{k+1}} \cdot W_{k+1}^{-1} \cos(N_k\theta^{(k)})\Big)\, ,
    \label{order_general}
\end{align}
where the normalization is $C_k=1_{n_{k-1}}\cdot W_{k-1}^{-1}1_{n_{k-1}} + 1_{n_{k+1}}\cdot W_{k+1}^{-1}1_{n_{k+1}}$ which corresponds to the weighted sum of  nodes and faces of the simplex, or the combined number of nodes and faces for unweighted simplicial complexes, see Appendix~\ref{simplicial_order} for details.

As expected, $R_k = 1$ if $\theta^{(k)}$ is in the harmonic space, which corresponds to full synchronization. 
Notice that for $k>1$, the harmonic space is in general not spanned by the constant vector, and full synchronization does not correspond to equal $\theta_j$ values on the $k$-simplices. 
The simplicial order parameter generalizes the notion of full synchronization to the instantaneous phase vector to be in the harmonic space, where the phases are in general not equal, except in the node Kuramoto case. 
This type of harmonic synchronization is therefore akin to a \textit{simplicial} phase locking, in which each higher-order phase evolves with a different proper frequency but overall the whole dynamics lives within the harmonic space, i.e. $\mathrm{ker}(L_k)$ for the corresponding $k$.
In addition, if the dimension of the harmonic space is larger than one, the fully synchronized state is in fact a linear combination of the basis vectors of the harmonic space. 
For $\alpha_2=0$ and if $\alpha_1$ is harmonic, i.e. $\alpha_1 \in \mathrm{Ker}(L_1)$, the particular linear combinations will be dictated by the choice of $\alpha_1$, or, if absent, by the choice of initial conditions. Thus our formulation extends the notion of full synchronization beyond constant phases to include a generalized harmonic phase lock.

As in the node Kuramoto case, this order parameter acts as a potential for the gradient flow formulation of the full $k$-order Kuramoto dynamics as 
\begin{align}
    \dot \theta^{(k)} = C_k W_k\nabla_{\theta^{(k)}} R_k^2(\theta^{(k)})\, . 
\end{align}
Note that this formulation does not contain the harmonic natural frequencies which can be recovered, as before, via a change of rotating frame.
Finally, we notice that in the case of the standard node Kuramoto, this measure corresponds to the weighted generalization of~\eqref{order_naive} with a different normalization factor
\begin{align}
    \mathcal R_0^2(\theta^{(0)}) = \frac{1}{C_0} 1_{n_1} \cdot W_1^{-1} \cos(N_1^*\theta^{(0)})\, ,
    \label{node_order}
\end{align}
where $C_0 =\sum_{i=0}^{n_1} (W_1^{-1})_{ii}$ is the weighted sum of edges, or for an unweighted graph, the number of edges.

\section{Examples}

\subsection{Frustrated simplicial Kuramoto model on a face}\label{sec:example_1}

\begin{figure*}[htbp]
    \centering
    \includegraphics[width=0.9\textwidth]{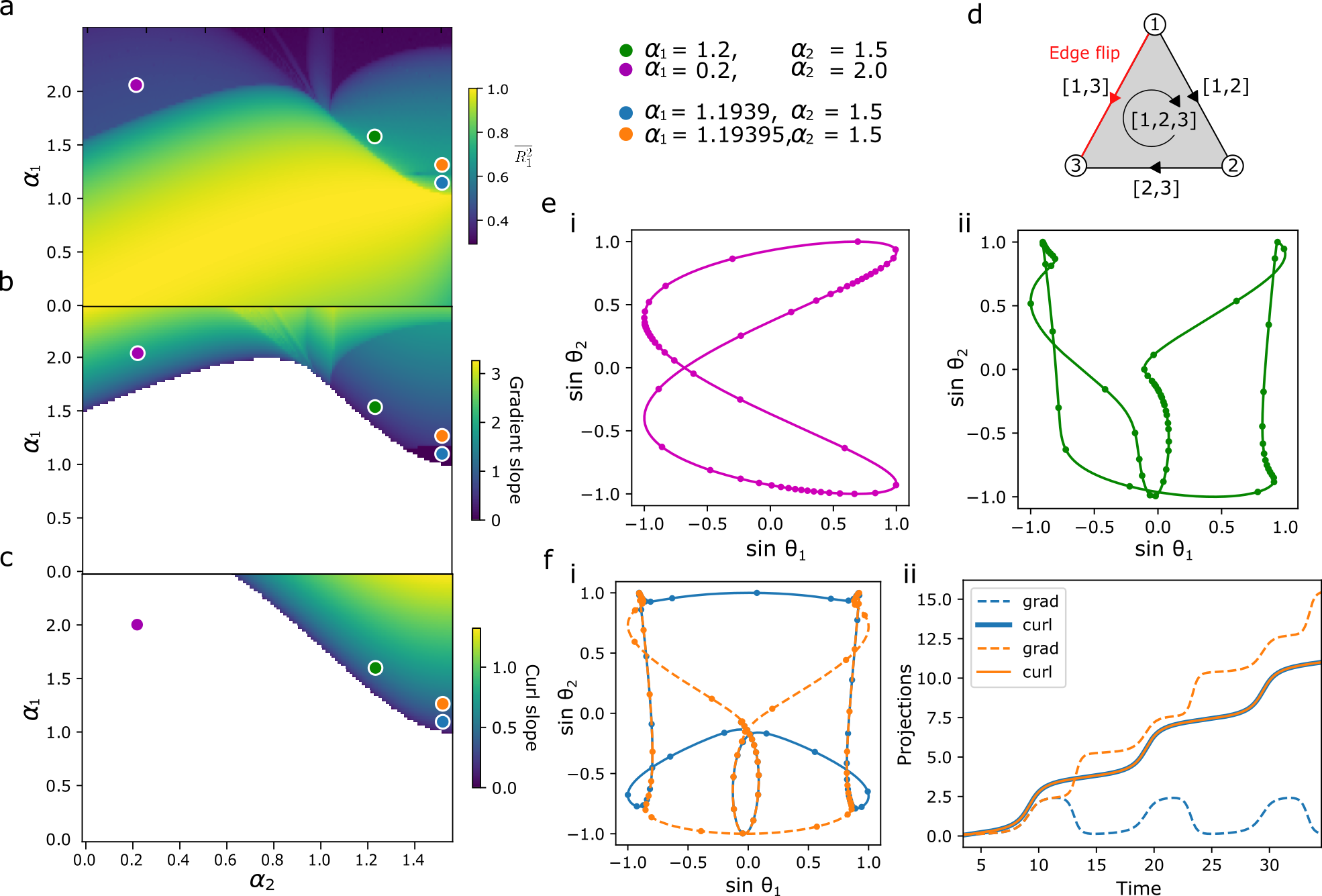}
    \caption{This figure illustrates the effect of frustration in a simplicial complex composed of a single face and with the orientation of an edge reversed (panel {\bf d}).
    We scan a range of values for both frustration parameters $\alpha_1$ and $\alpha_2$
    Panel {\bf a} shows the average value of the simplicial order parameter defined in~\eqref{order_general} in the stationary regime of the solution.
    Panels {\bf b} and {\bf c} respectively show the slope of the time evolution of the gradient and curl projection of the dynamics, see text for more details.
    The regions in white show where the projections are time independent, i.e. constant while dark blue - value of zero - are oscillating solutions around a fixed value. Higher values correspond to solutions that have a linear growth term in time.
    In panel {\bf e}, two typical stationary trajectories of the dynamics are shown in the regime with, panel {\bf e(i)}, and without, panel {\bf e(ii)}, a stationary curl. The frustration parameter values are indicated by the magenta and green dots respectively in panels {\bf a,b,c}.
    The circle markers on the trajectories are equally spaced in time along one cycle.
    Panel {\bf f} illustrates the sharp transition between vanishing and non vanishing slope of the projection of the gradient in {\bf f(ii)} and the resulting change of the trajectories in {\bf f(i)}. The frustration parameter values are shown by round markers of     corresponding colors in panels {\bf a,b,c}.
    Notice that both curl projections overlap across the transition in {\bf f(ii)}.
    }
    \label{fig:figure_1}
\end{figure*}

To showcase the properties of the frustrated simplicial Kuramoto model, we begin with one of the simplest examples: the single face complex of a triangle graph.
The single face triangle complex has no hole and thus the harmonic space of the Hodge Laplacian is of dimension zero. Therefore, to be in the full synchronization regime, defined as $\theta_i=\theta_j\ \forall i,j$ in the absence of harmonic space, one would expect that that state is only accessible for the non-frustrated model with $\alpha_1=\alpha_2=0$ in~\eqref{eq:edge_frustrated_kuramoto}. We will show it is not the case and the dynamics can still reach full-synchronization. 
For simplicity, and without loss of generality, we will use $\alpha_1$ as a constant vector in time with the same value on all three edges.
In addition, whilst the model is invariant to face orientation, it is not invariant to edge orientation. We thus have two non-equivalent choices for edge orientation: (i) a fully oriented complex, or (ii) one edge oriented in the opposite direction, as shown in Fig.~\ref{fig:figure_1}(d).

In (i) the fully oriented complex, all edges are equivalent and the frustrated simplicial Kuramoto model reduces to the scalar equation
\begin{align*}
    \dot \theta = -\alpha_1 - \sin(3\theta + \alpha_2)\, . 
\end{align*}
If $|\alpha_1| < 1$, any initial condition will converge, as time $\rightarrow\infty$ to full synchronisation with phase $\theta_\infty =\frac13 \left( \sin^{-1}(-\alpha_1) -\alpha_2\right)$ in the stationary state. 
Otherwise, the stationary solution will be periodic around a linearly increasing trend.
In (ii), the case of a flipped edge orientation, only two edges are equivalent, yielding the following coupled differential equations
\begin{align*}
    \dot \theta_1 &= -\alpha_1 - \sin(-\theta_1 + \theta_2) - \sin(2\theta_1 - \theta_2 + \alpha_2)\\
    \dot \theta_2 &= - \alpha_1 + 2\sin(-\theta_1 + \theta_2) - \sin(2\theta_1 - \theta_2 + \alpha_2)\, .
\end{align*}
We solve these equations numerically for values of $\alpha_1\in[0, 2.5]$ and $\alpha_2\in \left [0, \frac{\pi}{2}\right]$ and show in Fig.~\ref{fig:figure_1}(a-c) three measures that help us characterize the ensuing dynamics.

In Fig.~\ref{fig:figure_1}(a), we plot the simplicial order parameter ~\eqref{order_general} where we observe a full synchronization regime, $R^2_1(\theta^{(1)})=1$, for the non frustrated case with $\alpha_1=0,\ \alpha_2=0$, but also for a large region of the $\alpha_1$ and $\alpha_2$ parameter space. 
To understand this regime further in term of the Hodge decomposition of the stationary state, we show in Fig.~\ref{fig:figure_1}(b-c) the slope of a linear fit, representing the drift, of the temporal evolution of the projection of the solution onto the gradient and the curl subspaces at stationarity. 
More precisely, we estimate the parameter $a_h$, or slope, of the linear regression of the projections of $\theta(t)$ onto the grad, curl and harmonic spaces, as defined in \eqref{hodge_projections}. 
In addition, the white regions correspond to projections thata  are constant in time, while regions with vanishing slopes but non-constant projections are in dark blue - corresponding to a value of zero. The latter correspond to oscillating solutions around a constant value. As we will see below, these measures provide a finer, while still tractable, analysis of the solutions in the frustration parameter space than the order parameter.
We highlight some important observations from these plots. 
First, the region where the gradient component of the solution is non-constant matches with the region where we observe a large drop in synchronization. 
This suggests that when the gradient, and curl, component of the phases becomes too large, the synchronization is abruptly reduced. 
Notice that this region is bounded below by $\alpha_1=1$, as for the fully oriented case.
Second, the region where the projection of the curl is not constant is strictly contained within the region of non-constant gradient. This is a general result that we show below.

In Fig.~\eqref{fig:figure_1}(e) we show two typical trajectories of (i) non-constant gradient and (ii) non-constant gradient and non-constant curl (corresponding to the magenta and green markers on Fig.~\eqref{fig:figure_1}(a-c) respectively). 
We observe that the trajectory is a Lissajous curve when the curl component is constant and a more complex trajectory otherwise.
The Lissajous behavior is simply explained by imposing a constant curl, $\theta_2 = 2\theta_1 + \delta$ for a constant $\delta$, which reduces the coupled differential equations to a one dimensional dynamical system 
\begin{align*}
   \dot \theta_1 = -\alpha_1  - \sin(\alpha_2) - \sin(\theta_1)\, , 
\end{align*}
parameterizing the speed of motion on this curve, represented in Fig.~\eqref{fig:figure_1}(e)(i) as dots equally spaced in time.

Finally, in the regime with non-vanishing curl, upper right of Fig.~\eqref{fig:figure_1}(a-c), there exists a sharp transition along $\alpha_1$ between almost vanishing and positive gradient slope while the curl projection grows continuously.
In Fig.~\ref{fig:figure_1}(f), we show two trajectories on each side of this transition, corresponding to the blue (zero gradient slope) and orange (non-zero gradient slope) dots in Fig.~\eqref{fig:figure_1}(a-c). The two trajectories are partially overlapping where the segments of the trajectories parallel to the $\mathrm{sin}(\theta_1)$ axis are switching sign of $\mathrm{sin}(\theta_2)$. 
A more precise understanding of this transition in the context of dynamical system theory could be of interest but is beyond the scope of this work.

Although simple, this simplicial complex already displays interesting and non-trivial dynamical behavior of the simplicial Sakaguchi-Kuramoto model when the frustrations are turned on. 
However, this example does not contain a hole, i.e. there is no harmonic component to the dynamics. We explore the role of the harmonic component of the dynamic in the next section.

\subsection{Synchronization and edge orientation}\label{sec:example_2}

 \begin{figure*}[htbp]
 \centering
    \includegraphics[width=\textwidth]{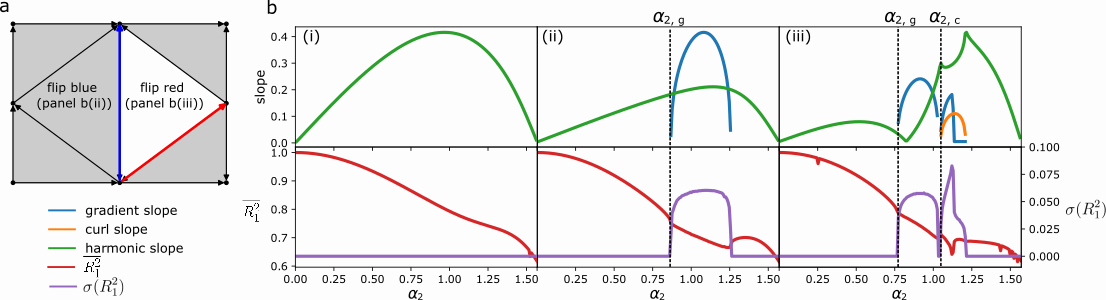}
    \caption{We consider the simplicial complex in panel {\bf a} comprising a single hole in white, faces in gray and edge orientation with black arrows. To study the effect of edge orientation on the dynamics, we construct two modified simplicial complexes with the (ii) blue edge reversed and (iii) both the blue and red edges reversed. 
    In {\bf b} we set the linear frustration parameter $\alpha_1=0$ and scan the nonlinear frustration parameter $\alpha_2$ for the three complexes and plot the  slope of the projection of the harmonic, gradient and curl component of the solution in the top row.
    If the slope value is absent, it corresponds to a constant projection. For example in panel {\bf bi}, only the harmonic projection is non-constant in time. If the slope is present with a value of $0$, the solution is oscillating around a fixed point, for example in panel {\bf biii} for gradient slope at large $\alpha_2$.
    In the bottom row, we show the average and standard deviation across time of the simplicial order parameter: $\overline{R_1^2}$ and $\sigma(R_1^2)$ respectively. The order parameter is $1$ for $\alpha_2=0$ and decreases as the nonlinear frustration increases. The standard deviation of the order parameter allows us to detect in which regime the solution is non-constant in the gradient or curl space. With this simplicial complex, the solution is non-constant only when the grad and/or curl are non-constant, even for high $\alpha_2$, which we will see is in contrast to the next example in Section~\ref{sec:example_3}.
    }
    \label{fig:figure_2}
\end{figure*}

For our second example, we use a slightly larger simplicial complex which we display in Fig.~\ref{fig:figure_2}(a) to study the properties of the dynamics in the presence of a hole.
We previously mentioned in Sections~\ref{sec:frustrated_kuramoto} and~\ref{sec:SOP}, that if $\alpha_1\in\mathrm{Ker}(L_1)$ and $\alpha_2=0$, the dynamics will fully synchronize with the stationary state $\theta^{(1)}=\alpha_1$. Setting $\alpha_2>0$ will perturb the stationary state by increasing the gradient and curl components, but may remain in a simplicial phase-lock for a wide range of parameters, including at high frustration. 

In Fig.~\ref{fig:figure_2}(b)(i), without loss of generality, we fix $\alpha_1=0$ and scan $\alpha_2\in \left [0, \frac{\pi}{2}\right ]$ and observe that for a given choice of edge orientation, the level of synchronization, as measured by the simplicial order parameter, decreases with $\alpha_2$, while its standard deviation remain null, which is an indicator of simplicial phase-locking.
The projections onto the gradient and the curl spaces are constant, while the harmonic projection is not.
We also notice that the dynamics are very sensitive to changes in orientation. Reversing the orientation of the blue edge, Fig.~\ref{fig:figure_2}(a), has a dramatic impact on the solution, with the gradient component becoming non-constant for some choices of $\alpha_2$, see Fig.~\ref{fig:figure_2}(b)(ii). Reversing the orientation of both the blue and red edges make both the gradient and curl components non-constant (see Fig.~\ref{fig:figure_2}(b)(iii)).
Similar to the first example of the single face triangle complex, we observe that the projection onto the curl is non-constant only if the projection onto the gradient is also non-constant.

We now show the existence of two critical values for $\alpha_2$ corresponding to changes of regime: $\alpha_{2, g}$ when the gradient becomes non-constant and $\alpha_{2, c}$ when the curl becomes non-constant, and that $\alpha_{2, c}\geq \alpha_{2, g}$.
For simplicity, we set $\alpha_1=0$, but any $\alpha_1\in\mathrm{Ker}(L_1)$ can be considered. 
First, for small $\alpha_2 < \alpha_{2, g}$, i.e. when the gradient and curl component are constant in the simplicial phase-lock regime, the solution is of the form
\begin{align*}
    \theta_{\infty}(t) = \Omega t h + \epsilon\, , 
\end{align*}
for a scalar $\Omega$, $h\in \mathrm{Ker}(L_1)$ and $\epsilon \in\mathrm{Ker}(L_1)^\perp $ small, and thus
\begin{align}
    \Omega h  = - L_1 \epsilon - (N^*_1V_2)^- \mathbf 1\alpha_2\, . 
    \label{omega-h}
\end{align}
The term $(N^*_k V_{k+1})^- 1_{n_{k+1}}$ counts the number of $k+1$-simplices adjacent to each $k$-simplex and is therefore a generalized degree. $(N^*_0 V_1)^- 1_{n_1}$ is simply the weighted node degree and $(N^*_1 V_2)^- 1_{n_2}$ the weighted edge degree.

From equation~\eqref{omega-h}, $\Omega$, $h$ and $\epsilon$ are defined as 
\begin{align}
    \Omega h &= - P_\mathrm{harm} (N^*_1V_2)^- \mathbf 1\alpha_2\label{Omega}\\
    L_1 \epsilon &= - P_\mathrm{harm}^\perp (N^*_1V_2)^- \mathbf 1\alpha_2\label{epsilon}\, . 
\end{align}
In this linear approximation, there always exists a solution of equation~\eqref{epsilon} for $\epsilon$, but, in the nonlinear regime, the presence of the sine function may prevent any solution to exist and the system will leave the phase-locked regime for $\alpha_2 > \alpha_{2, g}$.
The exact value of $\alpha_{2, g}$ is difficult to find analytically, as we see in Fig.~\eqref{fig:figure_2}, it depends not only on the structure of the simplicial complex but on the edge orientation as well.
In addition, if the dimension of the kernel of $L_1$ is larger than $1$, the direction of the vector $h$ in the harmonic space may also depend on $\alpha_2$.
For small $\alpha_2$, the value of $\Omega$ is represented in Fig.~\eqref{fig:figure_2}(b) by the value of the harmonic slope (in green), and increases quasi-linearly as a function of $\alpha_2$ as expected from equation~\eqref{Omega}.
For larger values of $\alpha_2$, the previous linearization is not valid and cannot be used to correctly approximate the dynamics.
However, the Hodge decomposition is still valid and the corresponding projections operators defined in equations~\eqref{hodge_projections} allow us to decompose the simplicial Kuramoto equation in its gradient and rotational parts as
\begin{align}
    \dot \theta_g := P_\mathrm{grad}\dot \theta^{(1)} &= - N_0 \sin(N_0^*\theta_g)\nonumber\\
     \qquad &- P_\mathrm{grad} (N_1^*V_2)^- \sin(V_2N_1\theta_c + \alpha_2)\label{eq:coupling}\\
    \dot \theta_c := P_\mathrm{curl} \dot \theta^{(1)} &= - P_\mathrm{curl} (N_1^*V_2)^- \sin(V_2N_1\theta_c + \alpha_2)\, , 
\end{align}
where $\theta^{(1)} = \theta_g + \theta_c + \theta_h$ from the Hodge decomposition. 
Notice that the Hodge decomposition directly implies that $N_1\theta^{(1)} = N_1\theta_c$ as $\theta_g$ is in the range of $N_0$, and $N_1N_0=0$, thus only the curl component of $\theta$, $\theta_c$, can survive in the sine terms containing $N_1$, a similar argument applies for $\theta_g$.
These two equations are coupled by the term $P_\mathrm{grad}(N_1^*V_2)^-\sin(V_2 N_1 \theta)$ which vanishes if $\alpha_2=0$ because $ P_\mathrm{grad}(N_1^*V_2)^- \sin(V_2 N_1 \theta) = P_\mathrm{grad}N_1^* \sin (N_1\theta)=0$, since the range of $N_1$ is orthogonal to the gradient space.
The fact that, in the absence of nonlinear frustration, the curl, grad and harmonic projection of the dynamics are decoupled was already noted in~\cite{millan2020explosive} and is a direct result of the orthogonality of these three spaces.
The nonlinear frustration makes the dynamics of the gradient projection depend on the solution of the curl projection.
This coupling relies on the presence of the lift and projections which are necessary to preserve the independence on the face orientation of the dynamics.

The presence of the coupling in the gradient equation explains why the dynamics of the curl projection can be non-constant only if the dynamics of the grad projection is also non-constant.
Indeed, for $\alpha_{2, g} < \alpha_2 < \alpha_{2, c}$, the curl projection equation is stationary with a time-independent $\theta_{c, \infty}$ solution, i.e. $\dot \theta_{c, \infty} = 0$.
The coupling term in the grad projection is then a constant, which we denote as $\delta$, and the Kuramoto dynamics reduces to 
\begin{align*}
    \dot \theta_g = - \delta -N_0 \sin(N_0^*\theta_g) \, . 
    \end{align*}
These dynamics correspond to the edge Kuramoto model on the complex without faces and a non-harmonic natural frequency $\delta$. It has a transition from synchronization to non-synchronization regime at $\alpha_{2, g}$.
For $\alpha_2 > \alpha_{2, c}$, the dynamics are non-stationary in the curl projection and in the gradient due to the presence of $\delta$.

While the order of the transitions to non-stationarity for the different components hold whenever they exist, their existence and exact behavior is dependent on the orientation of the edges and the localization of holes. 
Remarkably, even the simple example we used here displays an abundance of varying behaviors, of which we have only described representative examples: we observe no transitions in Fig.~\ref{fig:figure_2}(b)(i), only a gradient transition as in Fig.~\eqref{fig:figure_2}(b)(ii) or two transitions as in Fig.~\eqref{fig:figure_2}(b)(ii) with a near singular re-phase-locked synchronization gap.
In Fig.~\ref{fig:figure_2}(b)(ii-iii), we also observe a re-synchronization to a phase-locked regime for $\alpha_2>\alpha_{2, c}$ until $\frac{\pi}{2}$, possibly a result of the small size of the complex and high degree of symmetry. 
Indeed, as we observe in the next section, this regime does not exist for larger, more irregular complexes, see Fig.~\ref{fig:figure_3}, and is replaced by a more chaotic regime.

\subsection{Larger simplicial complex}\label{sec:example_3}

\begin{figure}[htpb]
    \centering
    \includegraphics[width=0.9\columnwidth]{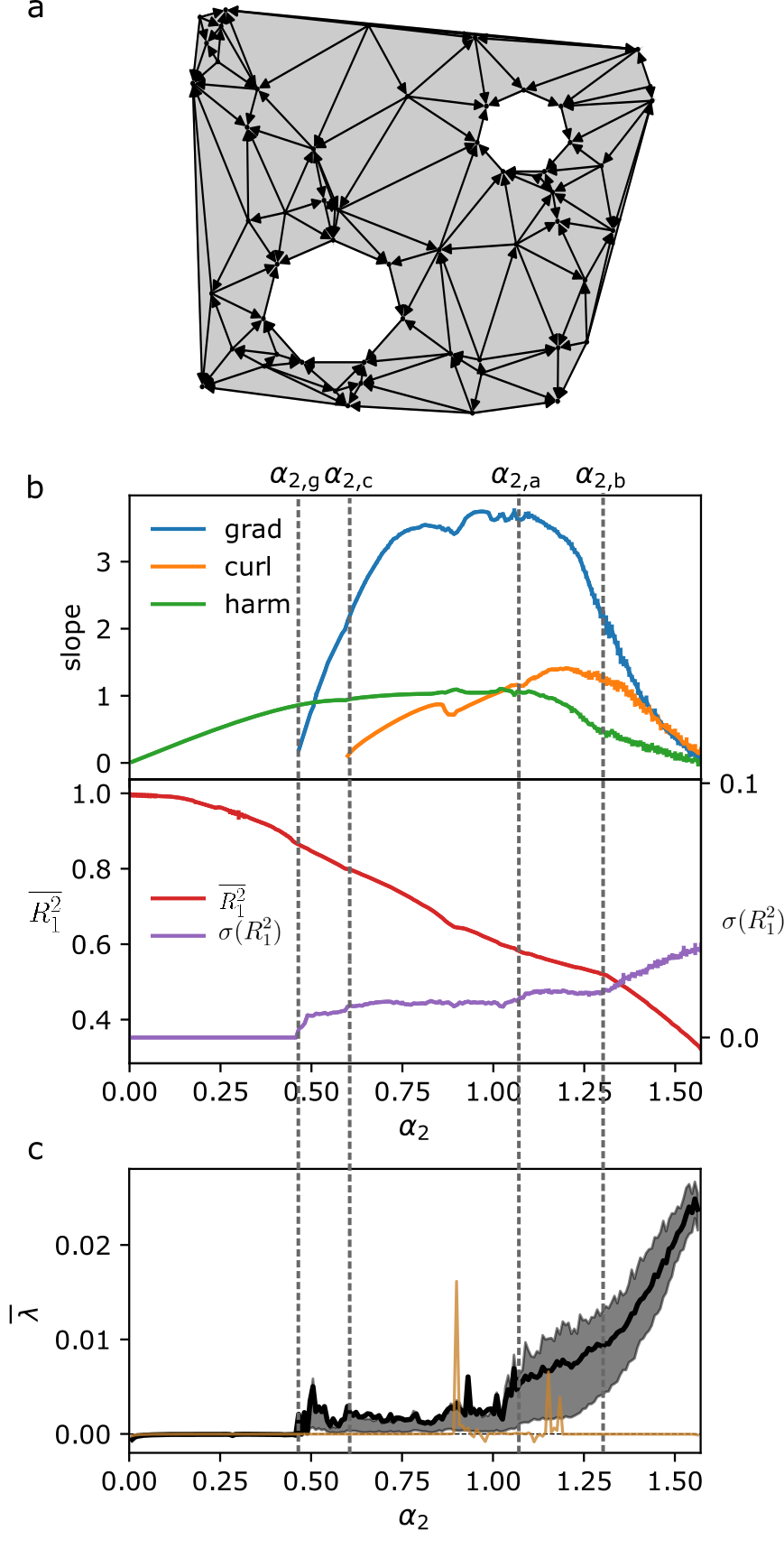}
    \caption{We consider the simplicial complex of panel {\bf a} obtained from a Delaunay triangulation of random points on a plane around two circular holes. In panel {\bf b}, we set the linear frustration parameter $\alpha_1=0$ and scan across the nonlinear frustration parameter $\alpha_2$. We plot the slope of the projection of the harmonic, gradient and curl component of the solution in the top row. If the slope value is absent, it corresponds to a constant projection. On the middle row, we show the average, $\overline{R_1^2}$, and standard deviation, $\sigma(R_1^2)$, across time of the simplicial order parameter $R_1^2$. In addition to the two critical $\alpha_2$, we manually highlighted two possible points corresponding to the onset of chaotic dynamics regimes with $\alpha_{2, a}$ and $\alpha_{2, b}$. 
    The standard deviation of the curves in the top of ${\bf b}$, which are calculated over $10$ simulations with random initial conditions, is zero, except for large $\alpha_2$ where it is small and is represented by the thickness of the curves.
    In {\bf c}, the dark line represents the average largest Lyapunov exponent $\overline{\lambda}$ over edges, and the shaded gray region between the lower an upper quartile of the corresponding distribution.
    For the sake of comparison, we have included a light brown curve  that is the mean largest Lyapunov exponent for the simplicial complex in Fig.~\ref{fig:figure_2}(b)(iii).
    }
    \label{fig:figure_3}
\end{figure}

Until now, we have studied the frustrated simplicial Kuramoto dynamics on small simplicial complexes in order to study and understand in detail the effects of the frustration on the dynamics.
As a final example for this paper, we consider a larger simplicial complex constructed from a Delaunay triangulation of random points on a plane around two circular holes as illustrated in Fig~\ref{fig:figure_3}(a).
In Fig.~\ref{fig:figure_3}(b), we show the same analysis as in Fig.~\ref{fig:figure_2} with the slope of the projections and the simplicial order parameter.
As expected, we observe more complex dynamics from the shape of these curves, obtained after averaging over $10$ simulations with random initial conditions.
In particular, we do not observe any re-synchronization for large $\alpha_2$ but rather an even more complex set of dynamics as shown by the standard deviation of the simplicial order parameter in Fig.~\eqref{fig:figure_3}(b).

To better quantify these complex dynamics, we compute the largest Lyapunov exponent~\cite{rosenstein1993practical, scholzel_christopher_2019_3814723} of the trajectories of each edge phase and show in Fig.~\ref{fig:figure_3}(c) the mean and quartile of them for each value of $\alpha_2$.
As soon as the dynamics are no longer constant, i.e. $\alpha_2>\alpha_{2, g}$, the largest Lyapunov exponent is on average positive, but increases significantly for larger $\alpha_2$, clearly indicative of chaotic dynamics. 
We visually noticed two different regime of chaotic dynamics, which we highlighted with $\alpha_{2, a}$ and $\alpha_{2, b}$ which corresponds to the start of the decrease of the slope of the gradient and the curl projection, respectively. 
For $\alpha_2>\alpha_{2, b}$, we also observe some sensitivity to initial condition on the value of the slope of the projection, the line thickness is the standard deviation across $10$ simulations with random initial conditions.

These two regimes would be interesting to study in more detail, since a decrease of the slope of the projection can either suggests more synchronization, as in the examples of Fig.~\ref{fig:figure_2}, or more random or chaotic dynamics, as in Fig.~\ref{fig:figure_3}. 
Understanding the transition between these two regimes in term of the complexity of the simplicial complex, where complexity is for example measured by the number of holes, their relative localization and the symmetries of the simplicial complex, is an open problem.
The Lyapunov exponent seems however a promising measure to identify the switching between the two regimes: for the example of Fig.~\ref{fig:figure_2}(biii) it remains at low values (see brown line in Fig.~\ref{fig:figure_3}c) and does not increase for large $\alpha_2$.

\subsection{Effect of weights}
\label{sec:weighted}
In this last example, we briefly explore the effect of weights on the Sakaguchi-Kuramoto dynamics using the general formulation of \eqref{eq:edge_frustrated_kuramoto}.
Our formulation allows mathematically for arbitrary weights on simplices of any order, however, the choice of weights may be constrained by the nature of the modelled system, e.g. geometric constraints of lengths, areas, and volumes. Whilst the exact interpretation of weights depends mainly on the context, we nevertheless provide general guidelines regarding their effect.
In the node Kuramoto model, weights on nodes can be interpreted as a modification of the underlying graph Laplacian and thus the dynamics it represents. For example, using the inverse degrees yields the normalised Laplacian. Edge weights are a natural mechanism to introduce heterogeneous interactions between oscillators and is known to give rise to interesting dynamics, e.g. metastable Chimera states~\cite{Shanahan:2010go}.
In the - frustrated - edge Kuramoto model~\eqref{frustrated_simplicial_kuramoto}, the focus of this paper, edge weights have a similar interpretation to edge weights in the node Kuramoto model: they quantify the strengths of the interactions for the node Kuramoto. Node and face weights both modulate the interaction strength between oscillators.
Face weights parameterize the strength of the triple coupling between phases where the nonlinear frustration acts, so vanishing faces weights are another mechanism to control the effect of the nonlinear frustration. Node weights play a similar role, but are decoupled from the effect of the frustration. We finally point out that while the weights at the different simplicial orders can related to each others, it is not necessarily the case, except in limit cases such as an edge with zero weights cannot serve as a support for a face.
A systematic exploration, both analytical and computational, of the role and effect of each type weights and their combinations is well beyond the scope of this paper. We present here a phenomenological description of the dynamics in a simple example as a preliminary to future work: we considered a simplicial complex comprised of two triangles, one full and one empty, sharing one face and varied the the weight $w$ of the full face, see Fig~\ref{fig:weighted_sc}a. 
In the limit where $w=0$, the face vanishes and we have two holes, and no frustration from $\alpha_2$ and the solution lies entirely in $\mathrm{Ker}(L_1)$. For low weights, only a small region is synchronized with a large increase in the projection in the gradient and harmonic spaces, but no curl component. The projection on the curl space is non-vanishing only for larger weights on the face. Finally for $w=1$, we recover a similar behavior to that of the simple triangle in section \ref{sec:example_1} and Fig.~\ref{fig:figure_1}.
The structure of the projection suggest non trivial relations between the different components, as well as three regimes: non-vanishing and vanishing curl, and two gradients regime nested in the vanishing curl one.
\begin{figure*}[htpb]
    \centering
    \includegraphics[width=0.8\textwidth]{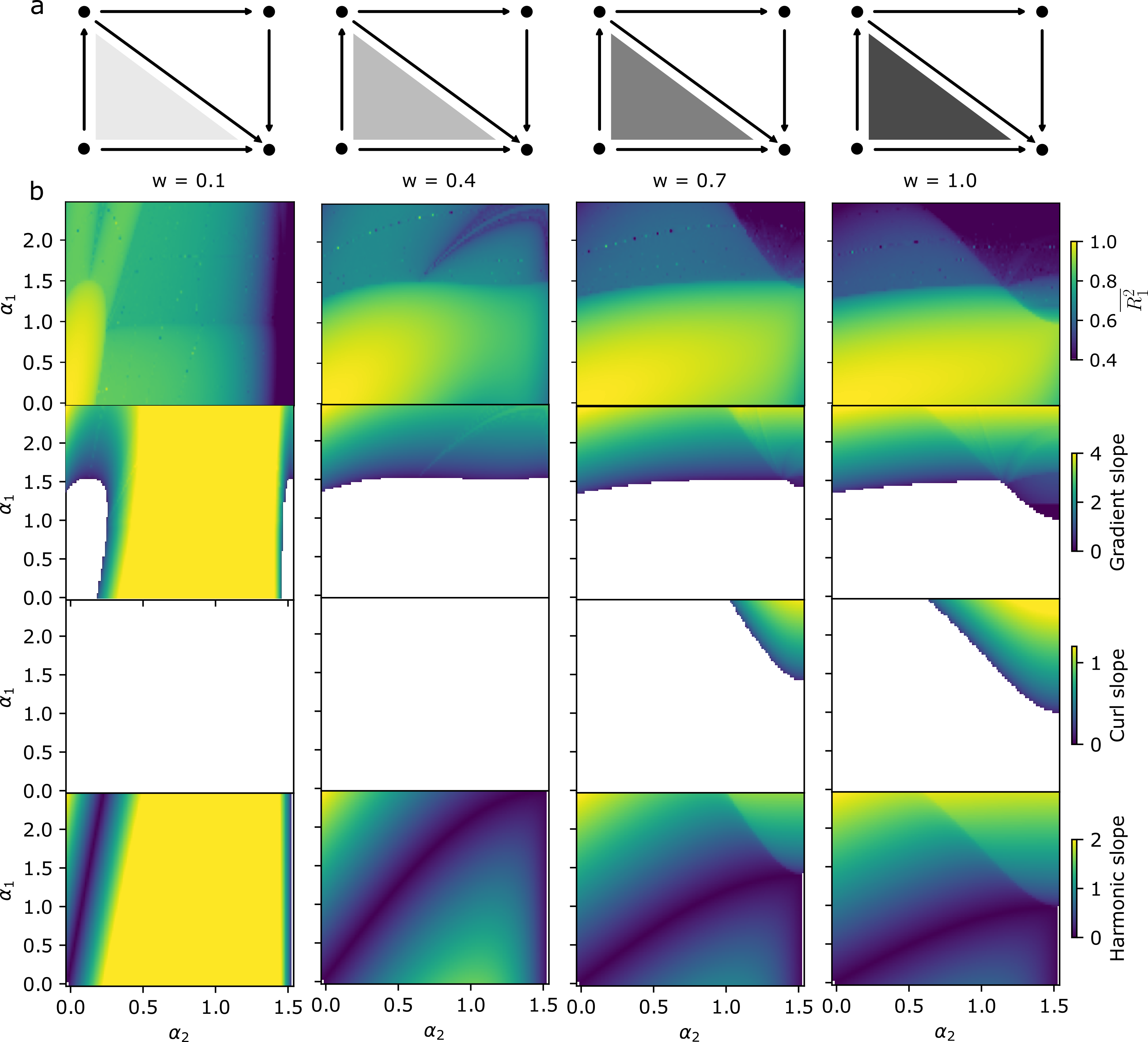}
    \caption{{\bf a} We consider a simplicial complex with one hole and one face parametrized by a weight $w\in[0,1]$.
    {\bf b} We scan the frustration parameters for four different values of $w>1$, as the case $w=1$ is trivial.  We display the order parameter as well as gradient, curl and harmonic slopes, as in Figure~\ref{fig:figure_1}. The structure of the projection suggest non trivial relations between the different components, as well as three regimes: non-vanishing and vanishing curl, and two gradients regime nested in the vanishing curl one.
    }
    \label{fig:weighted_sc}
\end{figure*}

\section{Conclusion}

In this work, we extend a previously introduced Kuramoto model on simplicial complexes~\cite{millan2020explosive,deville2021consensus} to include weights on any simplices as well as a linear and a non-linear frustration term to define the simplicial Sakaguchi-Kuramoto. 
This formulation naturally allows us to generalize the notion of synchronization, internal frequencies, and edge frustration, of the standard Kuramoto model.

Without frustration, the Kuramoto dynamics can be decomposed into three independent sub-systems aligned with the orthogonal spaces given by the Hodge decomposition.
However, we have demonstrated that by adding frustration to the dynamics, the harmonic, gradient and curl subspaces become hierarchically coupled, see equation~\eqref{eq:coupling}.
The dynamics in the harmonic space is coupled to both the gradient and curl subspaces even in the absence of harmonic linear frustration. In the linear regime of small nonlinear frustration, the amplitude of the dynamics in the harmonic space is proportional to the amount of frustration.
In the nonlinear regime of simplicial Sakaguchi-Kuramoto, we showed that the dynamics is highly varied, from constant to chaotic solutions.
Most surprisingly, the edge orientation is of fundamental importance in the resulting Kuramoto dynamics and the change of orientation of one edge can be enough to dramatically alter the dynamics.
Understanding the precise relationship between the choice of orientation for a given simplicial complex and the resulting type of dynamics has remained elusive so far but would be an interesting topic to gain further understanding of these systems particularly in the context of control.

We foresee various interesting directions for further interrogation of our frustrated simplicial Kuramoto and also additional adaptions. Firstly,while we only provide a simple example of the effect of weights on the dynamics, we believe that a full exploration of the weights definition and their effect will open interesting avenue of research not only theoretically but also for applications of the simplicial Sakaguchi-Kuramoto model. Secondly, whilst we used consensus dynamics in our formulation, we also mentioned earlier the dual formulation of the diffusion Kuramoto (see Appendix~\ref{diffusion_kuramoto}).  Indeed, examining how the dynamics of the consensus and diffusion formulations deviate in the weighted setting could be of interest. Thirdly, we did not explore the possible interplay between the linear and nonlinear frustration. 
The linear frustration is known to have a transition between stationary and non-stationary regimes~\cite{millan2020explosive}, but may also affect the types of dynamics with nonlinear frustration, which can even be made non-constant.
Finally, as we have shown with our three examples, the topology of the simplicial complex is crucial to determine the type of dynamics and in particular its complexity. 
A more complete characterisation in term of graph theoretical or topological measures would be of interest to identify the criteria necessary for the transition between non-stationary and chaotic dynamics.
In fact, particular geometries of simplicial complexes may support more specific types of dynamics, with maybe partial, cluster or metastable synchronisations.

\section*{Acknowledgments}

GP acknowledges partial support from Intesa Sanpaolo Innovation Center. The founder had no role in study design, data collection, and analysis, decision to publish, or preparation of the manuscript.
RP acknowledges funding through EPSRC award EP/N014529/1 supporting the EPSRC Centre for Mathematics of Precision Healthcare at Imperial and the Deutsche Forschungsgemeinschaft (DFG, German Research Foundation) Project-ID 424778381-TRR 295. PE acknowledges support from the NIHR Imperial Biomedical Research Centre (BRC) (grant number NIHR-BRC-P68711)
AA was supported by funding to the Blue Brain Project, a research center of the École polytechnique fédérale de Lausanne (EPFL), from the Swiss government’s ETH Board of the Swiss Federal Institutes of Technology.

\appendix

\section{Review of discrete geometry} \label{discrete_geo}

We provide here a short review of discrete geometry, following the exposition in~\cite{grady2010discrete} to which we redirect the reader for an in depth and detailed exposition of all the notions introduced here.
A simplicial complex is a collection of $k$-cells (node, edges, face, etc...), on which are defined $k$-chains as vectors with coefficients on each cell. 
The space of $k$-chains is denoted as $C_k$, and the incidence matrix $B_k^T$ maps $k+1$-chains to $k$-chains, i.e.
\begin{align*}
    B_k^T:C_{k+1}\to C_k\, .
\end{align*}

Dual to the space $C_k$ of $k$-chains is the space $C^k$ of $k$-cochains, defined using the scalar product as duality pairing, i.e. for $\tau_k\in C_k$ and $c^k \in C^k$, the pairing is 
\begin{align*}
    \langle \tau_k, c^k\rangle = \sum_i^{n_k} (\tau_k)_i (c^k)_i\, .
\end{align*}
From this pairing, the dual of the incidence matrix is defined as 
\begin{align*}
    \langle N_k^T\tau_{k+1}, c^k\rangle= \langle \tau_{k+1}, (N_k^T)^*c^k\rangle\, ,
\end{align*}
but reduces to coboundary operator 
\begin{align*}
    (N_k^T)^*=N_k:C^k \to C^{k+1}\, .
\end{align*}

We have thus defined the incidence matrix and its dual, but acting on chains and cochains. 
The map between them is a metric represented by a diagonal matrix, or weight matrix 
\begin{align*}
    W_k:C_k\to C^k
\end{align*}
as $c^k = W_k \tau_k$, and it's inverse 
\begin{align*}
    W_k^{-1}:C^k\to C_k\, .
\end{align*}
Then, to obtain the dual of the incidence matrix to form diffusion or Kuramoto equations, we need  to define the dual of simplicial complex and the Hodge operator.
If $n$ is the largest dimension of the $k$-cells in the complex, the dual of a complex is a complex where the $k$-cells of the primal complex are the $(n-k)$-cells of the dual complex.
With this definition, one can see that the dual incidence matrices $M_k$ are defined as $N_k^T = M_{n-k+1}$, that is the incidence matrix on $k$-cells is the transpose of the incidence matrix on $n-k+1$-cells in the dual complex.

The Hodge operator maps a $k$-cochain $\mathbf x$ of the primal complex to the $(n-k)$-cochain $\mathbf x^*$ on the dual as
\begin{align*}
    \mathbf x^* = \star \mathbf x := W_k^{-1}\mathbf x\, ,
\end{align*}
and the $(n-k)$-cochains $\mathbf y^*$ on the dual complex to the $k$-chain $\mathbf y$ on the primal complex as 
\begin{align*}
    \mathbf y = \star \mathbf y = W_k\mathbf y^*\, .
\end{align*}
We are now in the position to define $N_k^*$, the dual of the coboundary operator $N_k$, which maps $k+1$-cochains into $k$-cochains as 
\begin{align*}
    N_k^* = \star M_{n-k+1}\star = \star N_k^T\star= W_k N_k^TW_{k+1}^{-1}\, .
\end{align*}
Finally, the Hodge Laplacian is 
\begin{align*}
    L_k = N_{k-1}N_{k-1}^* + N_k^*N_k\, ,
\end{align*}
as defined in the main text.

\section{Lift, projection and frustrations}\label{lift_proj}

We provide here more details on the derivation of the frustrated simplicial Kuramoto model and the resulting orientation invariance. 
The most general projected lifted $L_k$ Laplacian with the lift from~\eqref{N_lift} and the projections onto negative values is 
\begin{align*}
    \widehat L_k^\mathrm{full} &= \frac12 (V_kN_{k-1}V_{k-1}^T)^-V_{k-1}N_{k-1}^*V_k^T \nonumber \\
    &+ \frac12 (V_kN_k^*V_{k+1}^T)^-V_{k+1}N_kV_k^T\,
\end{align*}
But we have the following propositions.
\begin{proposition}
The following holds for any $k$
\begin{align*}
    \frac12 V_k^T \widehat L_k^\mathrm{full} V_k = L_k \, .
\end{align*}
\end{proposition}

\begin{proof}
Using the relations $V_k^TV_k=2$ and $V_k^TV_k^-=1$
For the \textit{down} term of the Hodge Laplacian, we have
\begin{align*}
    V_k^T\widehat L_{k, down}^\mathrm{full}V_k &= V_k^T(V_k N_{k-1})^-N_{k-1}^*V_k^TV_k  \\
    &= V_k^TV_k^- N_{k-1}N_{k-1}^* 2 \\
    &= 2 N_{k-1}N_{k-1}^* \\
    &= 2 L_{k, down}\, , 
\end{align*}
which results in a non-lifted \textit{down} term in the complete Laplacian. 
For the up Laplacian, the same computation applies, thus we have the result.
\end{proof}
Then, from the form of the frustration operator acting in the nonlinear term corresponding to the \textit{up} Laplacian, we can only use the lift on the $k+1$ simplices to get the lifted Laplacian ~\eqref{laplacian-hat} of the main text.

\begin{proposition}
The frustrated simplicial Kuramoto model is independent on the orientation of the $k+1$ simplices.
\end{proposition}
\begin{proof}
The term of interest is
\begin{align*}
     L_\mathrm{up}(N_k, \theta^{(k)}) = (N_k^*V_{k+1})^- \mathrm{sin}\left (V_{k+1} N_k \theta^{(k)} + \alpha_{k+1}\right)\, .
\end{align*}
If we change the orientation of a $k+1$ simplex indexed by $i$ the corresponding coboundary operator $\widetilde N_j$ has $(\widetilde N_k)_i = -(N_k)_i$.
Then 
\begin{align*}
    V_{k+1}\widetilde N_k = P_iV_{k+1} N_k\, , 
\end{align*}
where $P_i$ permutes the rows $i$ and $2i$ of the lifted matrix.
Hence,  we have the orientation invariance
\begin{align*}
     L_\mathrm{up}(\widetilde N_k, \theta^{(k)}) &= (N_k^*V_{k+1}P_i)^- \mathrm{sin}\left (P_iV_{k+1} N_k \theta^{(k)} + \alpha_{k+1}\right)\\
     & = L_\mathrm{up}(N_k, \theta^{(k)}) \, , 
\end{align*}
as the permutation of rows commute with the point-wise sine function.
\end{proof}

\section{Simplicial order parameter}\label{simplicial_order}

Following~\cite{jadbabaie2004stability}, the generalisation of node order parameter to any graph~\eqref{order_naive} is obtained by rewriting the node order parameter for a complete graph with the graph incidence matrix
\begin{align*}
    n_0^2 R_{0, c}^2 &= \left |\sum_i \exp(i\theta_i)\right |^2\\
    &= \sum_i \exp(i\theta_i) \exp(-i\theta_i) + 2\sum_{i<j} \exp(i\theta_i) \exp(-i\theta_j) \\
    &= n + \sum_{i<j} \exp(i\theta_i) \exp(-i\theta_j) + \sum_{i>j} \exp(-i\theta_i) \exp(i\theta_j) \\
    &= n_0 + 2 1_{n_1}\cos(B_0\theta) \\
    &= n_0^2 - 2n_1 + 2 1_{n_1}\cos(B_0\theta) \\
    & = n_0^2 R_{0, g}^2
\end{align*}
where we used $n_1 = \frac12 (n_0^2-n_0)$, which holds for complete graphs.
We then chose a simpler normalization to write the order parameter as an average over edges as
\begin{align*}
     R_0^2 = \frac{1}{n_1} 1_{n_1}\cdot \cos(B_0\theta)\, , 
\end{align*}
with $1_{n_1}$ the constant vector of ones of dimension $n_1$, so that for full-synchronization, we still have $R_0^2 = 1$.

In order to obtain a proper generalization of the order parameter on simplicial complexes, one first has to notice that the order parameter generates the Kuramoto model as a gradient flow. 
A gradient flow is constructed from two elements: a convex potential function $H(x)$ for $x\in V$ with $V$ a vector space, and a gradient structure $K:V\to  V^*$ where $V^*$ is the dual of $V$. The gradient is a symmetric operator while an anti-symmetric operator would result in a Hamiltonian equation, with its dynamics restricted to the level sets of the potential function. 
The gradient flow is then given as 
\begin{align}
    \dot x =-K\frac{\delta H(x)}{\delta x}\, , 
\end{align}
where $\frac{\delta}{\delta x}: \mathcal F(V)\to V^*$ is a variational derivative acting on the space $\mathcal F(V)$ of functions of $V$ to its dual $V^*$. 

In our case, the potential function is $R_k^2$ defined in the main text in equation \eqref{order_general}, the variational derivative is simply the gradient with respect to $\theta$.
The gradient results in a chain, as
\begin{align*}
    W_{k-1}^{-1}N_{k-1}^*&: C^k \to C_{k-1} \quad \mathrm{and}\\
    W_{k+1}^{-1}N_k&: C^k \to C_{k+1}\, , 
\end{align*}
and the gradient structure $K$ is simply the weight matrix $W_k$ that converts the resulting chains to cochains.

\section{Diffusion simplicial Kuramoto}\label{diffusion_kuramoto}

All the derivations in this paper were done following the standard formulation of the Kuramoto model that corresponds to the consensus dynamics from the graph Laplacian.
This effect of this choice is actually only noticeable once one considers weighted graphs or simplicial complexes in a geometrical setting as presented here.
Indeed, without weights, both the diffusion and consensus dynamics are equivalent, as the graph Laplacian is symmetric.
Hence, there is an obvious formulation of Kuramoto model with the diffusion interpretation, that is obtained simply by acting with the Hodge Laplacian on the $\theta^{(k)}$ cochains from the right.
The diffusion simplicial Kuramoto model is then
\begin{align}
    \dot \theta^{(k)} = -\mathrm{sin}\left( \theta^{(k)}N_{k-1} \right)N_{k-1}^* - \mathrm{sin}\left(\theta^{(k)}N_k^* \right)N_k\, ,
    \label{diffusion_simplicial_kuramoto}
\end{align}
or, explicitly with the weight matrices,
\begin{align}
    \dot \theta^{(k)} &= -\mathrm{sin}\left (\theta^{(k)}B_{k-1} \right )W_{k-1} B_{k-1}^T W^{-1}_k\nonumber  \\
    &\qquad - \mathrm{sin}\left (\theta^{(k)}W_k B_k^T W_{k+1}^{-1}\right ) B_k \, ,
\end{align}
For $k=0$, the fully synchronized state is not the constant state as with the consensus dynamics but is proportional to $\mathrm{diag}(W_0^{-1})$, often given by the node degree for normalized graph Laplacian.

In order to include the frustration operator, one has to define it's 'dual' version acting on $N_k^*$ and we obtain the frustrated diffusion simplicial Kuramoto model 
\begin{align}
    \dot \theta^{(k)} = -\alpha_{k} -\mathrm{sin}\left( \theta^{(k)}N_{k-1} \right)N_{k-1}^* \nonumber \\
    - \mathrm{sin}\left(\theta^{(k)}N_k^* + \alpha_{k+1}\right)N_k\, .
    \label{frustrated_diffusion_simplicial_kuramoto}
\end{align}
Such systems require non-trivial weight matrices to be different from the simplicial Kuramoto model presented in the main text and could be of interest for future studies.

\section{Code availability}

The code to reproduce the figures is available on GitHub at \url{https://github.com/arnaudon/simplicial-kuramoto}.

\bibliographystyle{unsrt}
\bibliography{references}

\begin{thebibliography}{10}

\bibitem{Arenas:2008ku}
Alex Arenas, A~D{\'\i}az-Guilera, J{\"u}rgen Kurths, Y~Moreno, and C~Zhou.
\newblock {Synchronization in complex networks}.
\newblock {\em Physics Reports}, 2008.

\bibitem{Petri:2013bz}
Giovanni Petri, Paul Expert, Henrik~J Jensen, and John~W Polak.
\newblock {Entangled communities and spatial synchronization lead to
  criticality in urban traffic}.
\newblock {\em Scientific Reports}, 3:1--9, May 2013.

\bibitem{OKeefe:1971bj}
J~O'Keefe and J~Dostrovsky.
\newblock {The hippocampus as a spatial map. Preliminary evidence from unit
  activity in the freely-moving rat.}
\newblock {\em Brain research}, 34(1):171--175, November 1971.

\bibitem{Hafting:2005dp}
Torkel Hafting, Marianne Fyhn, Sturla Molden, May-Britt Moser, and Edvard~I
  Moser.
\newblock {Microstructure of a spatial map in the entorhinal cortex.}
\newblock {\em Nature}, 436(7052):801--806, August 2005.

\bibitem{kuramoto1975self}
Yoshiki Kuramoto.
\newblock Self-entrainment of a population of coupled non-linear oscillators.
\newblock In {\em International symposium on mathematical problems in
  theoretical physics}, pages 420--422. Springer, 1975.

\bibitem{acebron2005kuramoto}
Juan~A Acebr{\'o}n, Luis~L Bonilla, Conrad J~P{\'e}rez Vicente, F{\'e}lix
  Ritort, and Renato Spigler.
\newblock The kuramoto model: A simple paradigm for synchronization phenomena.
\newblock {\em Reviews of modern physics}, 77(1):137, 2005.

\bibitem{rodrigues2016kuramoto}
Francisco~A Rodrigues, Thomas K~DM Peron, Peng Ji, and J{\"u}rgen Kurths.
\newblock The kuramoto model in complex networks.
\newblock {\em Physics Reports}, 610:1--98, 2016.

\bibitem{Arenas:2006ba}
Alex Arenas, Albert Diaz-Guilera, and Conrad P{\'e}rez-Vicente.
\newblock {Synchronization Reveals Topological Scales in Complex Networks}.
\newblock {\em Physical Review Letters}, 96(11):114102--4, March 2006.

\bibitem{yeung1999time}
MK~Stephen Yeung and Steven~H Strogatz.
\newblock Time delay in the kuramoto model of coupled oscillators.
\newblock {\em Physical Review Letters}, 82(3):648, 1999.

\bibitem{Hellyer:2015ci}
Peter~J Hellyer, G~Scott, Murray Shanahan, D~J Sharp, and Robert Leech.
\newblock {Cognitive Flexibility through Metastable Neural Dynamics Is
  Disrupted by Damage to the Structural Connectome}.
\newblock {\em Journal of Neuroscience}, 35(24):9050--9063, June 2015.

\bibitem{hong2011kuramoto}
Hyunsuk Hong and Steven~H Strogatz.
\newblock Kuramoto model of coupled oscillators with positive and negative
  coupling parameters: An example of conformist and contrarian oscillators.
\newblock {\em Physical Review Letters}, 106(5):054102, 2011.

\bibitem{delabays2019kuramoto}
Robin Delabays, Philippe Jacquod, and Florian D\"orfler.
\newblock The kuramoto model on oriented and signed graphs.
\newblock {\em SIAM Journal on Applied Dynamical Systems}, 18(1):458--480,
  2019.

\bibitem{sakaguchi1986soluble}
Hidetsugu Sakaguchi and Yoshiki Kuramoto.
\newblock A soluble active rotater model showing phase transitions via mutual
  entertainment.
\newblock {\em Progress of Theoretical Physics}, 76(3):576--581, 1986.

\bibitem{Abrams:2004hq}
Daniel~M Abrams and Steven~H Strogatz.
\newblock {Chimera States for Coupled Oscillators}.
\newblock {\em Physical Review Letters}, 93(17):174102--4, October 2004.

\bibitem{Shanahan:2010go}
Murray Shanahan.
\newblock {Metastable chimera states in community-structured oscillator
  networks.}
\newblock {\em Chaos}, 20(1):013108, March 2010.

\bibitem{omel2012nonuniversal}
E~Omel’chenko and Matthias Wolfrum.
\newblock Nonuniversal transitions to synchrony in the sakaguchi-kuramoto
  model.
\newblock {\em Physical review letters}, 109(16):164101, 2012.

\bibitem{nicosia2013remote}
Vincenzo Nicosia, Miguel Valencia, Mario Chavez, Albert D{\'\i}az-Guilera, and
  Vito Latora.
\newblock Remote synchronization reveals network symmetries and functional
  modules.
\newblock {\em Physical review letters}, 110(17):174102, 2013.

\bibitem{wiesenfeld1996synchronization}
Kurt Wiesenfeld, Pere Colet, and Steven~H Strogatz.
\newblock Synchronization transitions in a disordered josephson series array.
\newblock {\em Physical review letters}, 76(3):404, 1996.

\bibitem{filatrella2008analysis}
Giovanni Filatrella, Arne~Hejde Nielsen, and Niels~Falsig Pedersen.
\newblock Analysis of a power grid using a kuramoto-like model.
\newblock {\em The European Physical Journal B}, 61(4):485--491, 2008.

\bibitem{Battiston:2020kp}
Federico Battiston, Giulia Cencetti, Iacopo Iacopini, Vito Latora, Maxime
  Lucas, Alice Patania, Jean-Gabriel Young, and Giovanni Petri.
\newblock {Networks beyond pairwise interactions: Structure and dynamics}.
\newblock {\em Physics Reports}, 874:1--92, August 2020.

\bibitem{iacopini2019simplicial}
Iacopo Iacopini, Giovanni Petri, Alain Barrat, and Vito Latora.
\newblock Simplicial models of social contagion.
\newblock {\em Nature communications}, 10(1):1--9, 2019.

\bibitem{Carletti:2020ux}
Timoteo Carletti, Duccio Fanelli, and Sara Nicoletti.
\newblock Dynamical systems on hypergraphs.
\newblock {\em Journal of Physics: Complexity}, 1(3):035006, 2020.

\bibitem{schaub2020random}
Michael~T Schaub, Austin~R Benson, Paul Horn, Gabor Lippner, and Ali Jadbabaie.
\newblock Random walks on simplicial complexes and the normalized hodge
  1-laplacian.
\newblock {\em SIAM Review}, 62(2):353--391, 2020.

\bibitem{millan2020explosive}
Ana~P Mill{\'a}n, Joaqu{\'\i}n~J Torres, and Ginestra Bianconi.
\newblock Explosive higher-order kuramoto dynamics on simplicial complexes.
\newblock {\em Physical Review Letters}, 124(21):218301, 2020.

\bibitem{DeVille_2021}
Lee DeVille.
\newblock Consensus on simplicial complexes: Results on stability and
  synchronization.
\newblock {\em Chaos: An Interdisciplinary Journal of Nonlinear Science},
  31(2):023137, 2021.

\bibitem{Ghorbanchian_2021}
Reza Ghorbanchian, Juan~G. Restrepo, Joaquín~J. Torres, and Ginestra Bianconi.
\newblock Higher-order simplicial synchronization of coupled topological
  signals.
\newblock 4:120, 2021.

\bibitem{Skardal:2019ik}
Per~Sebastian Skardal and Alex Arenas.
\newblock {Abrupt Desynchronization and Extensive Multistability in Globally
  Coupled Oscillator Simplexes}.
\newblock {\em Physical Review Letters}, 122(24):248301, June 2019.

\bibitem{Skardal:2020fl}
Per~Sebastian Skardal and Alex Arenas.
\newblock {Higher order interactions in complex networks of phase oscillators
  promote abrupt synchronization switching}.
\newblock {\em Communications Physics}, pages 1--6, November 2020.

\bibitem{calmon2021topological}
Lucille Calmon, Juan~G Restrepo, Joaqu{\'\i}n~J Torres, and Ginestra Bianconi.
\newblock Topological synchronization: explosive transition and rhythmic phase.
\newblock {\em arXiv preprint arXiv:2107.05107}, 2021.

\bibitem{grady2010discrete}
Leo~J Grady and Jonathan~R Polimeni.
\newblock {\em Discrete calculus: Applied analysis on graphs for computational
  science}.
\newblock Springer Science \& Business Media, 2010.

\bibitem{muhammad2006control}
Abubakr Muhammad and Magnus Egerstedt.
\newblock Control using higher order laplacians in network topologies.
\newblock In {\em Proc. of 17th International Symposium on Mathematical Theory
  of Networks and Systems}, pages 1024--1038. Citeseer, 2006.

\bibitem{jadbabaie2004stability}
Ali Jadbabaie, Nader Motee, and Mauricio Barahona.
\newblock On the stability of the kuramoto model of coupled nonlinear
  oscillators.
\newblock In {\em Proceedings of the 2004 American Control Conference},
  volume~5, pages 4296--4301. IEEE, 2004.

\bibitem{chapman2015advection}
Airlie Chapman.
\newblock Advection on graphs.
\newblock In {\em Semi-Autonomous Networks}, pages 3--16. Springer, 2015.

\bibitem{Eckmann:1945}
Beno Eckmann.
\newblock {Harmonische Funktionen und Randwertaufgaben in Einem Komplex}.
\newblock {\em Commentarii Math. Helvetici}, 17:240--245, 1945.

\bibitem{jiang2011statistical}
Xiaoye Jiang, Lek-Heng Lim, Yuan Yao, and Yinyu Ye.
\newblock Statistical ranking and combinatorial hodge theory.
\newblock {\em Mathematical Programming}, 127(1):203--244, 2011.

\bibitem{rosenstein1993practical}
Michael~T Rosenstein, James~J Collins, and Carlo~J De~Luca.
\newblock A practical method for calculating largest lyapunov exponents from
  small data sets.
\newblock {\em Physica D: Nonlinear Phenomena}, 65(1-2):117--134, 1993.

\bibitem{scholzel_christopher_2019_3814723}
Christopher Schölzel.
\newblock Nonlinear measures for dynamical systems, June 2019.

\end{thebibliography}

\end{document}